\documentclass[pra,aps,superscriptaddress,nofootinbib,twocolumn]
{revtex4-1}
\usepackage{etoolbox}
\apptocmd{\thebibliography}{}{}{}

\usepackage{mathrsfs}
\usepackage{amsfonts}
\usepackage{amssymb}
\usepackage{amsmath}
\usepackage{graphicx}

\usepackage{appendix}

\usepackage[braket]{qcircuit}
\usepackage{algorithm}
\usepackage[noend]{algpseudocode}
\usepackage{algpseudocode}
\usepackage{subfigure,dcolumn}
\usepackage{graphicx}
\usepackage{mathtools}
\usepackage{braket}
\usepackage{subfigure}
\usepackage{epstopdf}
\usepackage[usenames,dvipsnames]{color}
\usepackage[colorlinks=true,citecolor=blue,linkcolor=magenta]{hyperref}
\usepackage{ulem}
\usepackage{lmodern}
\usepackage{amsthm}

\newcommand{\Tr}{\mathrm{Tr}}

\usepackage{listings}
\usepackage{algorithm}  
\usepackage{algpseudocode}  
\usepackage{amsmath}

\usepackage[marginal]{footmisc}
\renewcommand{\thefootnote}
 
\newtheorem{lemma}{Lemma}

\begin{document}

\title{Classical post-processing approach for quantum amplitude estimation}

\author{Yongdan Yang}
\affiliation{Faculty of Science, Kunming University of Science and Technology, Kunming 650500, China}

\author{Ruyu Yang}
\email{yangruyu96@gmail.com}
\affiliation{Graduate School of China Academy of Engineering Physics, Beijing 100193, China}

\begin{abstract}

We propose an approach for quantum amplitude estimation (QAE) designed to enhance computational efficiency while minimizing the reliance on quantum resources. Our method leverages quantum computers to generate a sequence of signals, from which the quantum amplitude is inferred through classical post-processing techniques. Unlike traditional methods that use quantum phase estimation (QPE), which requires numerous controlled unitary operations and the quantum Fourier transform, our method avoids these complex and resource-demanding steps. By integrating quantum computing with classical post-processing techniques, our method significantly reduces the need for quantum gates and qubits, thus optimizing the utilization of quantum hardware. We present numerical simulations to validate the effectiveness of our method and provide a comprehensive analysis of its computational complexity and error. This hybrid strategy not only improves the practicality of QAE but also broadens its applicability in quantum computing.       

\end{abstract}


\maketitle

\section{Introduction}
Quantum computing holds the promise of significant computational advancements over traditional digital systems, particularly in specialized applications such as factorization \cite{shor1994algorithms} and unstructured database searches \cite{nielsen2001quantum}. Recently, many efforts have been dedicated to discovering new quantum algorithms to expand the application and potential of quantum computing. Among a variety of quantum algorithms, QAE plays a critical role due to its fundamental implications in enhancing the precision and efficiency of probability estimations in quantum states. Its applications are diverse, impacting fields ranging from quantum finance~\cite{rebentrost2018quantum,woerner2019quantum,egger2020credit,stamatopoulos2020option}, where it optimizes portfolio risks and derivative pricing, to quantum machine learning \cite{wiebe2015quantum,wiebe2016quantum,kerenidis2019q,kapoor2016quantum,li2019sublinear,miyahara2020quantum}, which benefits from improved pattern recognition and data classification. Additionally, QAE's influence extends to quantum chemistry \cite{knill2007optimal,kassal2008polynomial} for accurate molecular energy calculations and reaction rates, as well as solving complex optimization problems \cite{gacon2020quantum} and enhancing Monte Carlo simulations \cite{montanaro2015quantum,miyamoto2020reduction}.

Traditionally, QAE \cite{brassard2002quantum} combines QPE \cite{kitaev1995quantum} with the Grover search algorithm, requiring extensive quantum resources, including numerous ancillary qubits, controlled unitary operations and the quantum Fourier transform. These requirements make it unsuitable for early fault-tolerant quantum computers or noisy intermediate-scale quantum (NISQ) devices \cite{preskill2018quantum}. Recognizing these challenges, recent research has shifted towards developing more resource-efficient methods for conducting QAE that accommodate the constraints of existing quantum technologies.

According to Ref.~\cite{brassard2002quantum}, estimating the quantum amplitude can be transformed into estimating the phase of the corresponding amplification operator through QPE. Given the substantial research \cite{lin2022heisenberg,somma2019quantum,clinton2023phase,o2019quantum,russo2021evaluating,zhang2022computing} focused on reducing the quantum resource requirements for implementing the eigenvalue estimation algorithm, this naturally motivates the proposal of a resource-efficient approach for QAE. Notably, Suzuki et al \cite{suzuki2020amplitude} introduced an adaptation that uses fewer qubits and combines quantum operations with Maximum Likelihood Estimation (MLE) to circumvent the need for QPE, further streamlining the computational process. However, they did not provide a rigorous proof of their correctness. Similarly, the iterative quantum amplitude estimation (IQAE) \cite{grinko2021iterative} proposed by Grinko et al, employs a sequence of Grover-like iterations and classical statistical techniques to simplify the QAE process. These approaches reduce the quantum resource load by avoiding the full implementation of QPE, making them well-suited to the capabilities of near-term hardware technology. Additionally, there are other works \cite{ aaronson2020quantum,giurgica2022low,callison2022improved,rall2023amplitude,labib2024quantum,plekhanov2022variational} aimed at improving QAE algorithms. Our method proposed in this paper can be regarded as an extension of these algorithms.

In this paper, we introduce a classical post-processing approach for QAE. In our method, we transform the problem of quantum amplitude estimation into solving the phase of a specific amplification operator. We use the quantum computer to generate a series of classical signals containing amplitude information, which can then be extracted through classical post-processing, such as the classical Fourier transform. Our method is a resource-efficient algorithm that does not require the implementation of controlled time evolution. To demonstrate the efficacy of our proposed algorithm, we performed calculations on 4-qubit state quantum amplitude estimation and 6-qubit observable estimation. Moreover, we identify three distinct sources of error in our algorithm: (i) truncation of the evolution time, (ii) stochastic noise arising from finite sampling (shot noise), and (iii) circuit-level noise caused by hardware imperfections. Numerical simulations demonstrate the robustness of our algorithm against shot noise arising from finite measurements. This analysis provides insights into the algorithm's performance under a variety of conditions.

This paper is structured as follows: Sec. \uppercase\expandafter{\romannumeral2} introduces the preliminary theory of QAE. Sec. \uppercase\expandafter{\romannumeral3} provides a numerical simulation of our approach. In Sec. \uppercase\expandafter{\romannumeral4}, we analyze the algorithmic error. Sec. \uppercase\expandafter{\romannumeral5} concludes the paper.

\section{Algorithm}
 
In this section, we introduce some preliminary theories of QAE. Given two non-orthogonal states, $\ket{\psi}$ and $\ket{\phi}$, our ultimate goal is to estimate the quantum amplitude $|\braket{\psi|\phi}|^2$. The aim of QAE is to achieve the highest precision for $|\braket{\psi|\phi}|^2$ with the lowest computational expense. To accomplish this, we propose a classical post-processing approach that enhances the precision of QAE while avoiding the use of QPE. Our approach consists of two main steps: constructing the amplification operator and estimating the quantum amplitude without QPE. Additionally, we will demonstrate that our algorithm does not even require the Hadamard test. The specific process is outlined below.

\begin{figure}[t]
\centering\includegraphics[width=0.5\hsize]{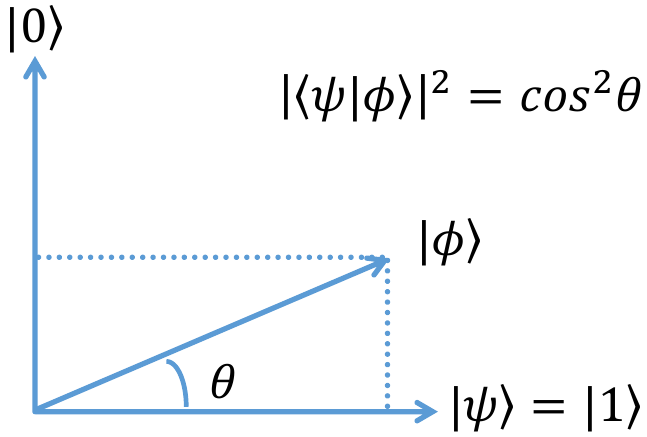}
\caption{The two-dimensional subspace spanned by states $\ket{\psi}$ and $\ket{\phi}$.}
\label{fig:sub}
\end{figure}

\subsection{Construct the amplification operator}

In this section, we introduce how to construct the amplification operator. The process is inspired by the work \cite{bravyi2011schrieffer}. We define two reflection operators as
\begin{equation} \label{reflec}
\begin{aligned}
&R_{\psi}=I-2\ket{\psi}\bra{\psi}      \\
&R_{\phi}=I-2\ket{\phi}\bra{\phi}.
\end{aligned}
\end{equation}
\begin{lemma}
Let $\mathscr{H}$ be a two-dimensional subspace spanned by states $\ket{\psi}$ and $\ket{\phi}$ shown in Fig.~\ref{fig:sub}. A unitary rotation in $\mathscr{H}$ from $\ket{\psi}$ to $\ket{\phi}$ can be written as: 
\begin{equation}
 U=\sqrt{R_{\phi}R_{\psi}}=exp(i\theta\sigma^{y})   
\end{equation}
Note that $\textcolor{red}{\sigma^y}$ is a Pauli operator defined in the $\mathscr{H}$, and $\theta$ is the angle between $\ket{\psi}$ and $\ket{\phi}$.  
\end{lemma}
\begin{proof}
Since $\ket{\psi}$ and $\ket{\phi}$ belong to the two-dimensional subspace $\mathscr{H}$, without loss of generality, we can set
\begin{equation}
\begin{aligned}
&\ket{\psi}=\ket{1} \\
&\ket{\phi}=cos(\theta)\ket{1}+sin(\theta)\ket{0},
\end{aligned}
\end{equation}
where $\ket{0}$ and $\ket{1}$ are the eigenstates of Pauli operator $\sigma^z$ defined in the subspace $\mathscr{H}$ . Using the definition in Eq.~\eqref{reflec}, we can get
\begin{equation}
\begin{aligned}
&R_{\psi}=I-2\ket{1}\bra{1}=\sigma^z \\
&R_{\phi}=cos(2\theta)\sigma^z-sin(2\theta)\sigma^x.
\end{aligned}
\end{equation}
Therefore, we can obtain the rotation operator as
\begin{equation}
\begin{aligned}
U=\sqrt{R_{\phi}R_{\psi}}=\sqrt{cos(2\theta)I+isin(2\theta)\sigma^y}=exp(i\theta\sigma^{y})
\end{aligned}
\end{equation}
\end{proof}
To implement the reflection operator in a system of $n$ qubits, as defined in Eq.\eqref{reflec}, we combine the state preparation operators $U_{\psi}$ and $U_{\phi}$, where $\ket{\psi}=U_{\psi}\ket{0}^{\otimes n}$ and $\ket{\phi}=U_{\phi}\ket{0}^{\otimes n}$, with the reflection operator $I-2(\ket{0}\bra{0})^{\otimes n}$. The latter is realized through a multi-controlled NOT gate, along with single-qubit Pauli gates and Hadamard gates \cite{nielsen2001quantum, chowdhury2018improved}. Consequently, the two reflection operators in Eq.\eqref{reflec} can be decomposed as follows
\begin{equation} \label{reflec1}
\begin{aligned}
&R_{\psi}=U_{\psi}(I-2(\ket{0}\bra{0})^{\otimes n})U_{\psi}^{\dagger}      \\
&R_{\phi}=U_{\phi}(I-2(\ket{0}\bra{0})^{\otimes n})U_{\phi}^{\dagger}.
\end{aligned}
\end{equation}

Based on Lemma 1, we can derive the equation about  quantum amplitude
\begin{equation}
\begin{aligned}
|\braket{\psi|\phi}|^2=cos^2(\theta),
\end{aligned}
\end{equation}
where assuming $0 \leq \theta \leq \pi/2$. Therefore, the task of estimating quantum amplitude can be transformed to estimate the value of angle $\theta$. Then, we can construct the amplitude amplification operator as
\begin{equation}
\begin{aligned} \label{7}
\mathscr{A}=U^2=exp(i2\theta\sigma^{y}).
\end{aligned}
\end{equation}
Given an initial state 
\begin{equation}
\begin{aligned} \label{ini}
\ket{\psi_0}=\frac{1}{\sqrt{2}}(\ket{y_+}+\ket{y_-}),
\end{aligned}
\end{equation}
where $\ket{y_\pm}$, corresponding to eigenvalues $\pm 1$, represent the eigenstates of Pauli operator $\sigma^{y}$ in the subspace$\mathscr{H}$, the process of amplitude amplification can be formulated as
\begin{equation} \label{inphase}
\begin{aligned}
\mathscr{A}^m \ket{\psi_0}=\frac{1}{\sqrt{2}}(exp(i2m\theta)\ket{y_+}+exp(-i2m\theta)\ket{y_-}). 
\end{aligned}
\end{equation}
Here, $m$ denotes the times of applying the amplification operator. According to Eq.~\eqref{inphase}, we have amplified the angle from $\theta$ to $2m\theta$, and the information of angle $2m\theta$ has been encoded into the phase of the final state $\mathscr{A}^m \ket{\psi_0}$. This process allows us to attain high-precision amplitude values. To estimate the quantum amplitude, we can extract the phase information of the final state using our approach, which is shown below.

\subsection{Estimating the quantum amplitude without QPE}

\begin{figure}[t]
\centering\includegraphics[width=1\hsize]{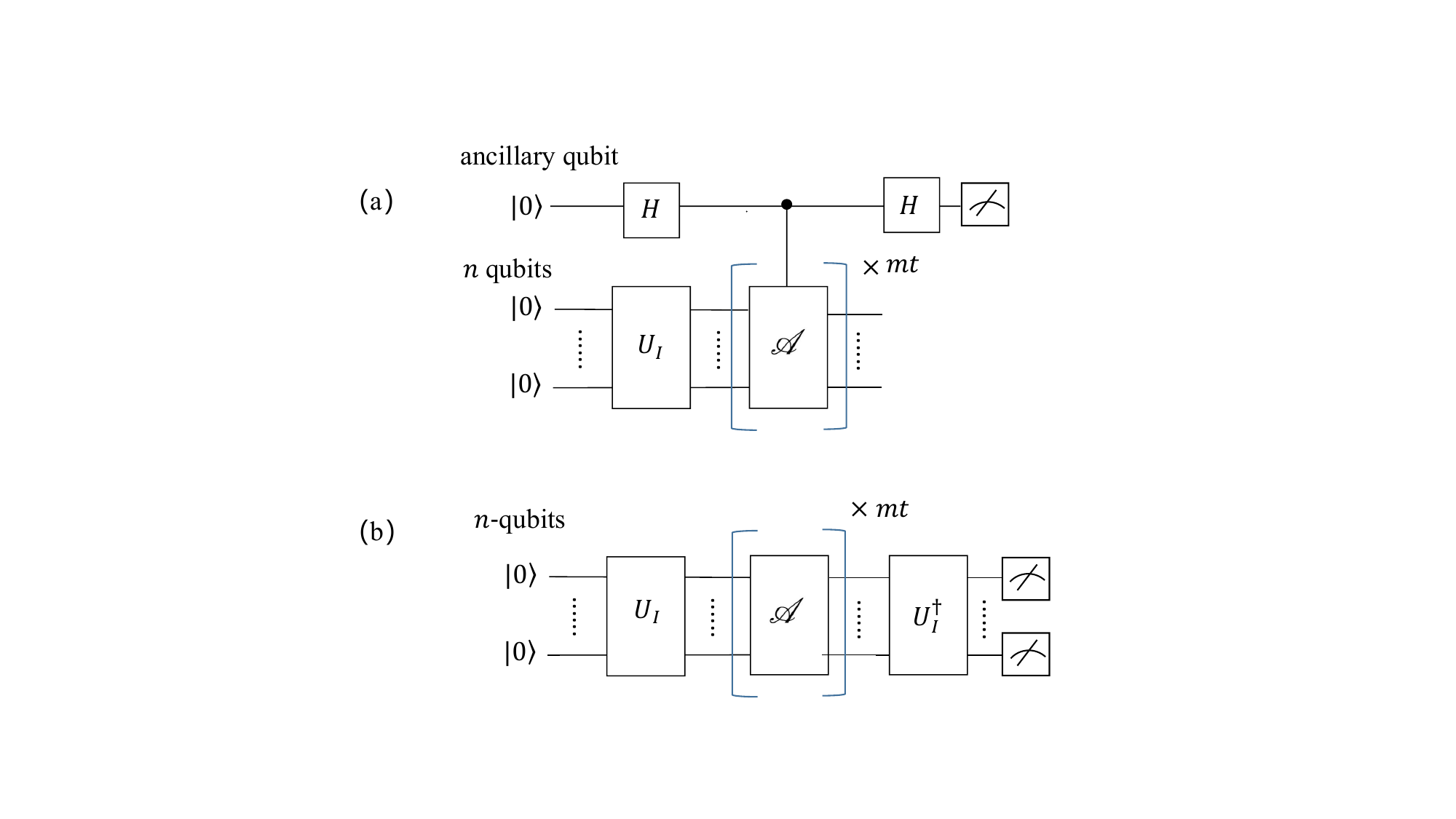}
\caption{(a) The quantum circuit for estimating the overlap $\bra{\psi_0}\mathscr{A}^{mt}\ket{\psi_0}$ with Hadamard test, where $n$ denotes the number of qubits required to represent the quantum state and $\ket{\psi_0}$ is the initial state satisfying $\ket{\psi_0}=U_I \ket{0}^{\otimes n}$. (b) The quantum circuit for estimating the overlap $|\bra{\psi_0}\mathscr{A}^{mt}\ket{\psi_0}|^2$ without Hadamard test.}
\label{fig:cir}
\end{figure}

To extract the information of the angle $\theta$ from the amplification operator $\mathscr{A}$ defined in Eq.~\eqref{7}, we consider $S(x)$ as the Fourier transform of function $f(t)$
\begin{align} 
S(x)=\sum_{t=-\infty}^{\infty} f(t)e^{i x t},
\label{eq:s(x)}
\end{align}
with $f(t)$ a function as
\begin{equation}
\begin{aligned}
f(t)&=p(t)\Tr[\mathscr{A}^{mt}\ket{\psi_0}\bra{\psi_0}]\\
&=p(t)\bra{\psi_0}\mathscr{A}^{mt}\ket{\psi_0}, 
\label{f}
\end{aligned}
\end{equation}
where $\ket{\psi_0}$ is the initial state and $\mathscr{A}^{mt}$ means applying the amplification operator $\mathscr{A}$ with $mt$ times. Note that the variables $t$ and $m$ must be integer. Here, $p(t)$ is a cooling function adopted from Ref.~\cite{zeng2021universal,yang2024resource} to guarantee that $S(x)$ can be efficiently evaluated within a finite range. Since it is shown that the Gaussian function has the best performance among several other example functions, in this work, we assume that $p(t)$ is 
the Gaussian function as an example in the further analysis, i.e., 
\begin{align}
p(t)=e^{-a^2t^2},
\end{align}
where $0<a<1$ is a tuning parameter. Therefore, the Eq.~\eqref{eq:s(x)} can be rewritten as
\begin{equation} 
\begin{aligned} \label{eq:sx}
S(x)&=\sum_{t=-\infty}^\infty e^{-a^2t^2}\Tr[\mathscr{A}^{mt}\ket{\psi_0}\bra{\psi_0}]e^{i x t}  \\
&=\sum_{t=-\infty}^\infty e^{-a^2t^2}\bra{\psi_0}\mathscr{A}^{mt}\ket{\psi_0}e^{i x t}  \\
&=\frac{\sqrt{\pi}}{2a}(e^{\frac{-(x+2m\theta)^2}{4a^2}}+e^{\frac{-(x-2m\theta)^2}{4a^2}}).
\end{aligned}
\end{equation}
The specific derivation process of Eq.~\eqref{eq:sx} is shown in Appendix~\ref{sec A}. By modifying the value of $x$, $S(x)$ exhibits a local maximum or peak when $x=\pm 2m\theta$. Hence, we can extract the information of the value of $\theta$. The value of the overlap $\bra{\psi_0}\mathscr{A}^{mt}\ket{\psi_0}$ in Eq.~\eqref{eq:sx} can be measured in a quantum computer by the Hadamard test. The quantum circuit is shown in Fig.~\ref{fig:cir}(a). It is worth noting that by increasing the depth of the quantum circuit, we can perform amplitude amplification, thereby scaling the angle $\theta$ to $2m\theta$. This approach enhances the precision of estimating $\theta$ with the same number of measurements, at the cost of increasing the circuit depth by a factor of $m$.

In our method, we utilize the summation method to directly compute $S(x)$ defined in the Eq.~\eqref{eq:s(x)}. During the process, we choose a finite summation range $[-T, T]$ instead of the infinite summation range $[-\infty,\infty]$. While the summation range of $t$ is changed from $[-\infty,\infty]$ to $[-T, T]$, it will induce a cutoff error, $\epsilon\sim \mathcal O(e^{-(aT)^2})$. 
The detailed error analysis can be found in the Appendix \ref{sec: cut}.

\subsection{Estimating the quantum amplitude without Hadamard test}

In certain cases, the Hadamard test can become a hindrance to further development, especially when the quantum circuit complexity increases or when noise becomes significant. The overhead introduced by the Hadamard test can degrade computational efficiency. In this section, we introduce how to estimate the quantum amplitude without the Hadamard test. The quantum circuit is shown in Fig.~\ref{fig:cir}(b). We set the function $f(t)$ mentioned in Eq.~\eqref{eq:s(x)} as 
\begin{equation}
\begin{aligned}
 f(t)&=e^{-a^2t^2}\Tr[\ket{\psi_0}\bra{\psi_0}\mathscr{A}^{mt}\ket{\psi_0}\bra{\psi_0}\mathscr{A}^{-mt}]\\
&=e^{-a^2t^2}|\bra{\psi_0}\mathscr{A}^{mt}\ket{\psi_0}|^2,  
\end{aligned}
\end{equation}
Where $\mathscr{A}^{-mt}$ denotes the inverse of $\mathscr{A}^{mt}$. Then, the Eq.~\eqref{eq:s(x)} can be written as
\begin{equation} 
\begin{aligned} \label{eq:sx2}
S(x)&=\sum_{t=-\infty}^\infty p(t)|\bra{\psi_0}\mathscr{A}^{mt}\ket{\psi_0}|^2 e^{i x t}  \\
&=\frac{\sqrt{\pi}}{4a}(e^{\frac{-(x+4m\theta)^2}{4a^2}}+e^{\frac{-(x-4m\theta)^2}{4a^2}}+2e^{\frac{-x^2}{4a^2}}).
\end{aligned}
\end{equation}
Here, the initial state is defined by Eq.~\eqref{ini}. By modifying the value of $x$, $S(x)$ exhibits a local maximum or peak when $x=\pm 4m\theta$ or $0$. Hence, we can extract the information of the value of $\theta$. Note that the value of the term $|\bra{\psi_0}\mathscr{A}^{mt}\ket{\psi_0}|^2$ in the Eq.~\eqref{eq:sx2} can be obtained by measuring the final state $\mathscr{A}^{mt}\ket{\psi_0}$ directly without Hadamard test. However, compared to the method mentioned in Section B, this method introduces additional peaks at $x = 2 \pi N$, where $N = 0, 1, 2, \dots$, which might interfere with the true peaks at $x = 4m\theta$.

\subsection{Application to observable estimation}

Our method can also be used to estimate the observable, e.g., $B=\sum_i c_i P_i$. Here, the coefficient $c_i$ is a real number, and $P_i$ denotes the Pauli operator string. We aim to estimate the average value $\bra{\psi} P_i \ket{\psi}$. According to the Lemma 1, we can set 
\begin{equation}
\begin{aligned} \label{14}
&\ket{\psi}=\ket{1} \\
&\ket{\phi}=P_i \ket{\psi} = \cos(\theta)\ket{1}+\sin(\theta)\ket{0}.
\end{aligned}
\end{equation}
Note that the Pauli operator string $P_i$ can also be regarded as a unitary operator. So the equation $\ket{\phi}=P_i \ket{\psi} = \cos(\theta)\ket{1}+\sin(\theta)\ket{0}$ is valid. Based on the Eq.~\eqref{14}, we can derive the equation 
\begin{equation}
\begin{aligned} \label{15}
\bra{\psi} P_i \ket{\psi}=\cos(\theta),
\end{aligned}
\end{equation}
where assuming $0 \leq \theta \leq \pi$. Therefore, the task of estimating the average value of the Pauli operator string $\bra{\psi} P_i \ket{\psi}$ can be transformed into estimating the value of angle $\theta$. Then, similar to the Eq.~\eqref{7} in section A, we can construct the amplitude amplification operator as
\begin{equation}
\begin{aligned}
\mathscr{A}=U^2=exp(i2\theta\sigma^{y}).
\end{aligned}
\end{equation}
The value of the angle $\theta$ can be obtained using the method described in Sections B and C.

\section{numerical simulation}

For numerical simulation, we can summarize our method as follows
\begin{enumerate}
    \item Select an initial state $\ket{\psi_0}$ which has a reasonable overlap with the eigenstates $\ket{y_{\pm}}$ in the subspace $\mathscr{H}$.
    
    \item Choose an appropriate summation range $[-T,T]$ as defined in Eq.~\eqref{eq:s(x)} and set a specific magnification factor, i.e., $m=M$. 
    
    \item Evolve the selected initial state $\ket{\psi_0}$ with amplification operator $\mathscr{A}^{mt}$ in a quantum computer, where $t$ must be a integer and range from $-T$ to $T$. 
    
    \item Perform measurements to obtain $\bra{\psi_0}\mathscr{A}^{mt}\ket{\psi_0}$ or $|\bra{\psi_0}\mathscr{A}^{mt}\ket{\psi_0}|^2$, and multiply it by $e^{ixt}$ for the parameter $x$. 
    
    \item Compute the summation $S(x)=\sum_{t=-T}^{T} f(t)e^{i x t}$ by a classical computer.
\end{enumerate}

\subsection{Estimate the quantum amplitude}

\begin{figure}[t]
\centering\includegraphics[width=0.95\hsize]{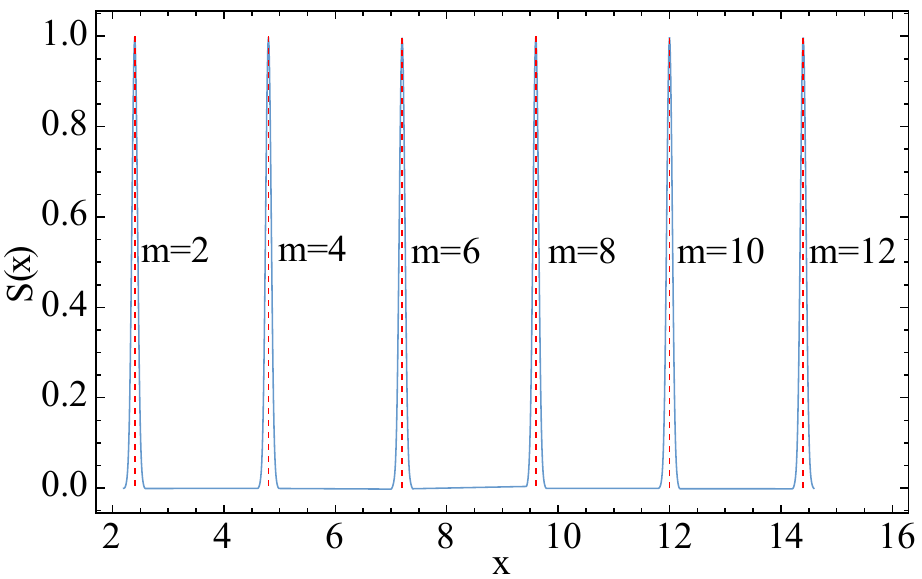}
\caption{The quantum amplitude of two 4-qubit state $\ket{\psi}$ and $\ket{\phi}$. In the figure, the red vertical line indicates the exact value of the $2 m \theta$, while the light blue line represents $S(x)$ as evaluated with our method. }
\label{fig:4q-qae}
\end{figure}

To demonstrate the efficiency of our approach, we choose two 4-qubit states $\ket{\psi}$ and $\ket{\phi}$ as an example and utilize our method to estimate the quantum amplitude $|\braket{\psi|\phi}|^2=cos^2(\theta)$, where our main task is to estimate the value of angle $\theta$. Here, we utilize the Eq.~\eqref{eq:sx} to extract the information of angle $\theta$. Therefore, the Hadamard test is necessary. 

The specific process to evaluate $S(x)$ is shown above and the final results are shown in Fig.~\ref{fig:4q-qae}. Here, the parameter $a$ is fixed as $\frac{1}{20 \sqrt{2}}$. We use a cutoff of the sum range from $[-\infty,\infty]$ to $[-60, 60]$ when evaluating $S(x)$ in Eq.~\eqref{eq:sx}. We set the exact value of the angle as $\theta=0.6$ and change the magnification factor $m$ from 2 to 12. As shown in Fig.~\ref{fig:4q-qae}, the location of peaks
matches well with the exact value of angle $2 m \theta$.

To further demonstrate the robustness of our method, we randomly select six pairs of 4-qubit states, $\ket{\psi}$ and $\ket{\phi}$, for quantum amplitude estimation. The corresponding results are shown in Fig.\ref{fig:4q-random}. In this experiment, the magnification factor is fixed at $m=10$, while the parameter $a$ is set to $\frac{1}{20 \sqrt{2}}$. We apply a cutoff in the summation range from $[-\infty, \infty]$ to $[-60, 60]$ when evaluating $S(x)$ in Eq.\eqref{eq:sx}. As shown in the figure, the location of the peaks in $S(x)$ closely matches the exact value of $2 m \theta$, demonstrating the accuracy and reliability of our method in the context of randomly selected quantum states.

\begin{figure}[t]
\centering\includegraphics[width=0.95\hsize]{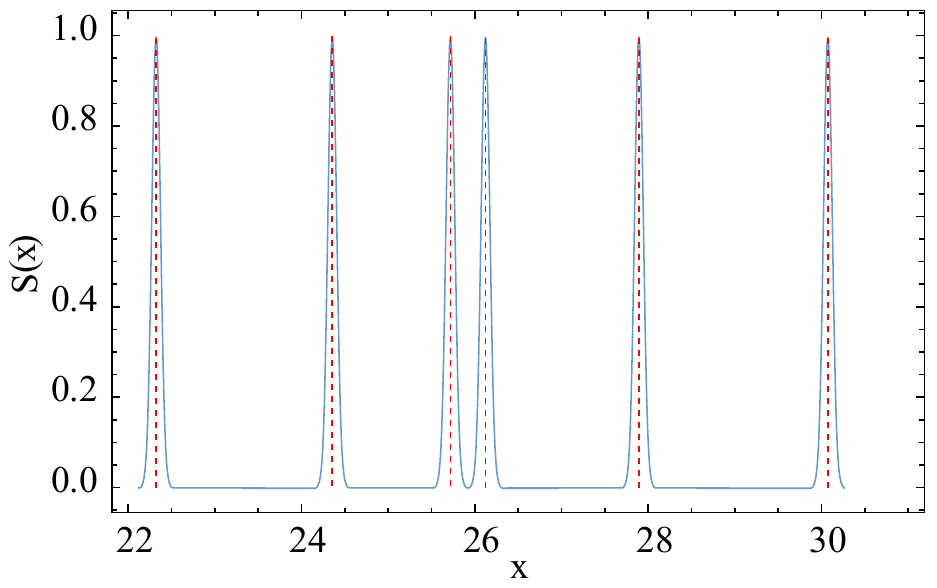}
\caption{The quantum amplitude of six pairs of randomly selected 4-qubit states $\ket{\psi}$ and $\ket{\phi}$. In the figure, the red vertical line indicates the exact value of the $2 m \theta$, while the light blue line represents $S(x)$ as evaluated using our method.}
\label{fig:4q-random}
\end{figure}

\subsection{Estimate the observable}

Our method can also be used to estimate the average value of an observable such as the Pauli string. According to the theory mentioned in Section \uppercase\expandafter{\romannumeral2}, we can transform the task of estimating an observable into estimating the value of angle $\theta$ as shown in Eq.~\eqref{15} i.e., $\bra{\psi} P_i \ket{\psi}=cos(\theta)$, where $P_i$ denotes the Pauli string. Here, we use the Eq.~\eqref{eq:sx2} to evaluate the function $S(x)$. Hence, we need to measure the term $|\bra{\psi_0}\mathscr{A}^{mt}\ket{\psi_0}|^2$ without implementing the Hadamard test.

The specific process of estimating $\theta$ is similar to the process of QAE mentioned above, and the final results are shown in Fig.~\ref{fig:6q-pauli}. Here, the parameter $a$ is fixed as $\frac{1}{20 \sqrt{2}}$. We use a cutoff of the summation range from $[-\infty,\infty]$ to $[-60, 60]$ when evaluating $S(x)$ in Eq.~\eqref{eq:sx2}. We set the exact value of the angle as $\theta=0.595$ and change the magnification factor $m$ from -1 to 14. As shown in Fig.~\ref{fig:6q-pauli}, the location of peaks matches well with the exact value of angle $4 m \theta$.  

\begin{figure}[t]
\centering\includegraphics[width=0.95\hsize]{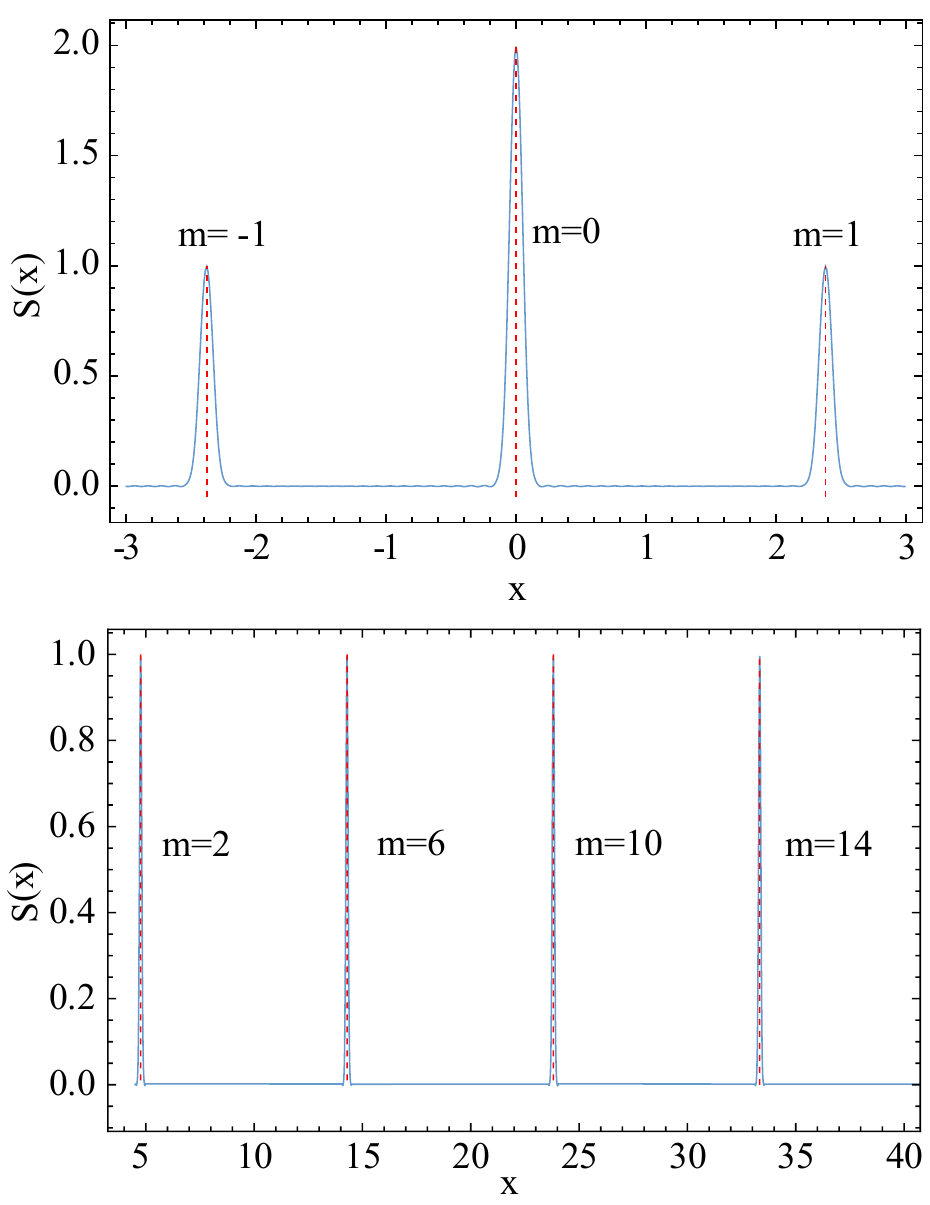}
\caption{Estimating the average value of the 6-qubit Pauli string $I_1 I_2 I_3 I_4 Z_5 Z_6$ under the  state $\ket{\psi}$. The red vertical line indicates the exact value of the $4 m \theta$, while the light blue line represents $S(x)$ as evaluated with our method.}
\label{fig:6q-pauli}
\end{figure}

\section{error analysis}
The core of our algorithm is to utilize Eq.~\eqref{eq:s(x)} to extract the information of quantum amplitude. During the process, the error of our algorithm mainly arises from \textcolor{red}{three aspects}: (1) To limit the evolution time, we use a cutoff for the summation range from $[-\infty,\infty]$ to $[-T, T]$ when evaluating $S(x)$ Eq.~\eqref{eq:s(x)}. (2) Shot noise is caused by the finite number of measurements when we measure the overlap $\bra{\psi_0}\mathscr{A}^{mt}\ket{\psi_0}$ or $|\bra{\psi_0}\mathscr{A}^{mt}\ket{\psi_0}|^2$ on a quantum computer. (3) Circuit-level noise. Next, we will discuss the three aspects separately.

\subsection{Error from cutoff}

We define $S(x)^\infty$ and $S(x)^T$ as the values obtained from an infinite summation range and a cutoff range, respectively. The cutoff error $\epsilon = |S(x)^\infty - S(x)^T|$ can be shown to satisfy $\epsilon \leqslant \frac{2}{a} e^{-a^2 T^2}$. The detailed error analysis can be found in Appendix \ref{sec: cut}. Therefore, if we want to constrain the error $\epsilon$ to $\epsilon_c$, the summation range should be no smaller than $\frac{1}{a}\sqrt{\ln\left(\frac{2}{a \epsilon_c}\right)}$. In other words, the error caused by the cutoff will decay exponentially as we increase the cutoff range. It is important to note that although the cutoff affects the values of $S(x)$ itself, in our algorithm, this does not affect the location of its peak. Hence, it does not impact the accuracy of the angle $\theta$ derived from $S(x)$. This aspect will be the focus of our subsequent discussion.

Here, we use Eq.\eqref{eq:sx} as the definition of $S(x)$. We choose the initial state as $\ket{y_-}$ (choosing $\ket{y_+}$ would be equivalent) and set the summation range to $[-T, T]$. Then, Eq.\eqref{eq:sx} can be written as: 
\begin{equation} \label{sx_ct} 
\begin{aligned} 
&S(x) = \sum_{t=-T}^T \bra{y_-}\mathscr{A}^{mt}\ket{y_-} e^{-a^2 t^2} e^{i x t} \\ 
&= 1+2e^{-a^2} \cos(x - 2m\theta)+2e^{-4a^2} \cos2(x - 2m\theta)+.... 
\end{aligned} 
\end{equation}
According to Eq.\eqref{sx_ct}, when $x = 2m\theta$, $S(x)$ exhibits a local maximum or peak. Particularly, when the cutoff range is set to $[-T, T] = [-1, 1]$, which represents the smallest possible range, Eq.\eqref{sx_ct} can be simplified to: \begin{equation} \label{sx_small}
S(x) = 1 + 2 e^{-a^2} \cos(x - 2m\theta). \end{equation} 
$S(x)$ can still exhibit a peak at $x = 2m\theta$. This means that, even when the cutoff range is reduced to its minimum, the accuracy of the final result remains unaffected. Next, we take the 4-qubit state as an example and estimate the quantum amplitude at different cutoff ranges. The final results are shown in Fig.~\ref{fig:4q-error}.

\begin{figure}[t] 
\centering\includegraphics[width=0.95\hsize]{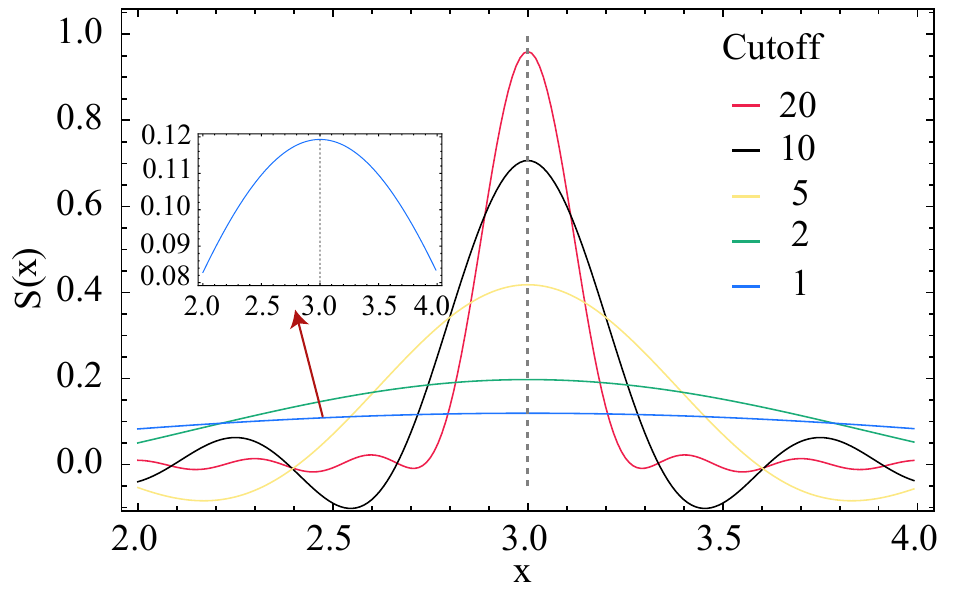}
\caption{The curves of $S(x)$ correspond to different cutoff ranges which are captured by different colors. The dashed vertical line marks the exact location of the $2 m \theta$}
\label{fig:4q-error}
\end{figure}

In Fig.~\ref{fig:4q-error}, we illustrate the influence of the cutoff range. We set the magnification factor as $m=1$. The initial state is chosen as $\ket{y_-}$. We set the trunction range from $[-20,20]$ to $[-1,1]$. The parameter $a$ is fixed to be $\frac{1}{10 \sqrt{2}}$.  As shown in the figure, $S(x)$ reaches a peak at different truncation range when $x=2m\theta$. It can be concluded that the truncation doesn't change the correct position of the peak but will reduce the resolution of the curve. 

\subsection{Error from shot noise}

\begin{figure*}[t] 
\centering\includegraphics[width=0.85\hsize]{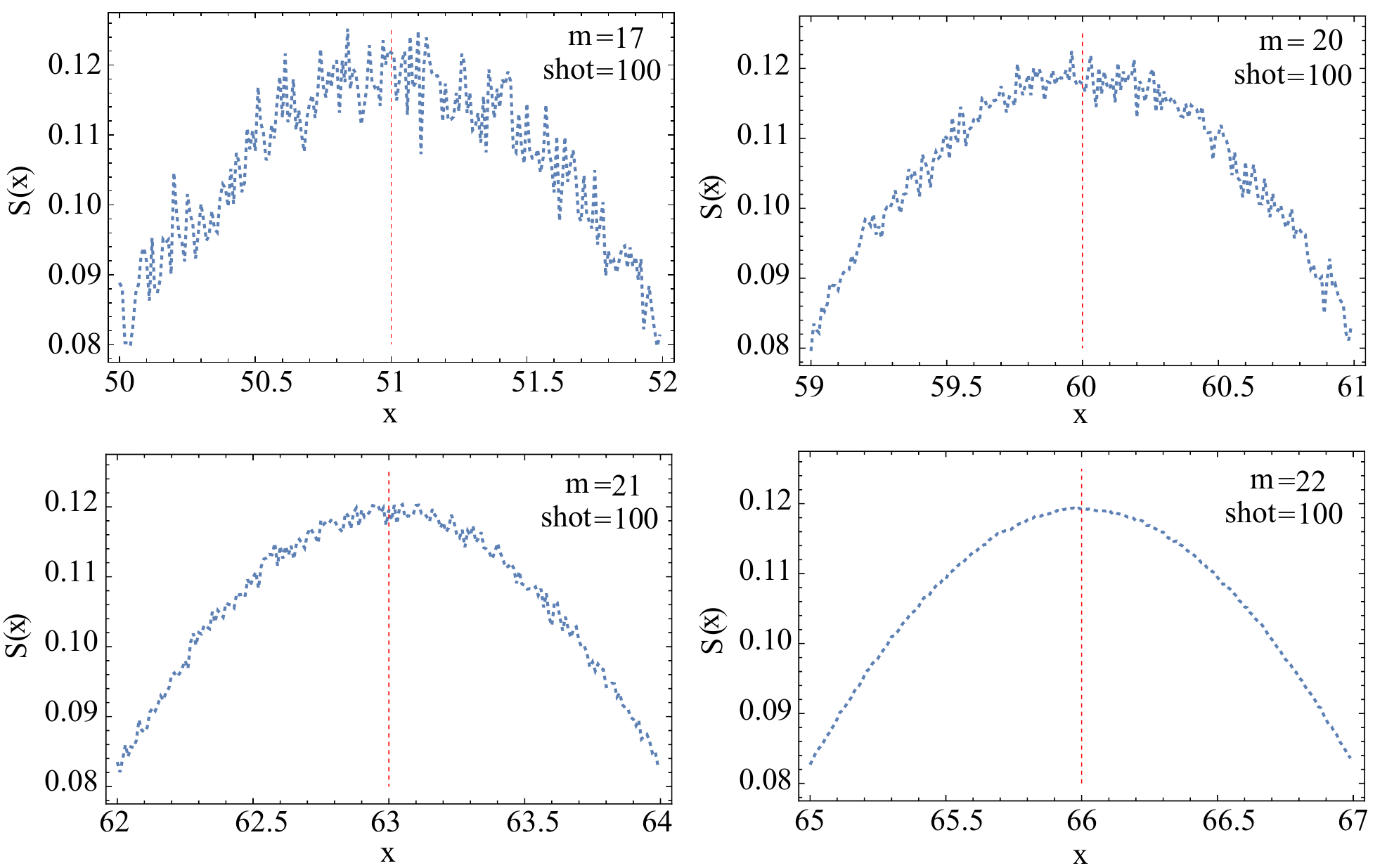}
\caption{The shot noise analysis of 4-qubit state's amplitude estimation. The dashed vertical line marks the exact location, while the light blue line represents $S(x)$ as evaluated with our method. The different figures are obtained by setting the different magnification factors $m$ with the same shot numbers.}
\label{fig:4q-shot noise}
\end{figure*}

In our algorithm, the most important step is to measure the overlap $\bra{\psi_0}\mathscr{A}^{mt}\ket{\psi_0}$ on a quantum computer. Due to the finite number of quantum measurements, the process will inevitably induce the shot noise. Here, we can write
\begin{equation} \label{sx_small} 
\begin{aligned}
\bra{\psi_0}\mathscr{A}^{mt}\ket{\psi_0}=\alpha+i\beta,\end{aligned} 
\end{equation}
where $\alpha$ and $\beta$ represent the real part and imaginary part of the overlap respectively. According to the principle of quantum measurement, the variance of the estimator $\hat{\alpha}$ and $\hat{\beta}$ can be formulated as
\begin{equation} \label{var} 
\begin{aligned}
Var(\hat{\alpha})&=\frac{(1-\alpha^2)}{N_{shot}}\\
Var(\hat{\beta})&=\frac{(1-\beta^2)}{N_{shot}},
\end{aligned} 
\end{equation}
where $N_{shot}$ represents the number of measurements. According to Eq.~\eqref{var}, when $|\alpha|$ or $|\beta|$ is close to 1, the variance reaches the minimum value. When we set the initial state as $\ket{y_-}$, Eq.~\eqref{sx_small} can be written as
\begin{equation} \label{overlap} 
\begin{aligned}
\bra{y_-}\mathscr{A}^{mt}\ket{y_-}=\cos(2mt\theta)-i\sin(2mt\theta)=\alpha-i\beta.
\end{aligned} 
\end{equation}
According to the Eq.~\eqref{var} and Eq.~\eqref{overlap}, when $\alpha$ satisfy the following equation 
\begin{equation} \label{condition} 
\begin{aligned}
\alpha=\cos(2mt\theta)=\pm 1, 
\end{aligned} 
\end{equation}
the variance of the estimator $\hat{\alpha}$ becomes minimal, that is, zero. Therefore, in practice, as long as we measure $\alpha$, we can infer 
$\beta$. Thus, when the variance of 
$\alpha$ decreases, the variance of 
$\beta$ will also decrease. Based on this, we can conclude that the variance of the overlap $\bra{y_-}\mathscr{A}^{mt}\ket{y_-}$ becomes minimal when the magnification factor satisfies 
\begin{equation} \label{mg} 
\begin{aligned}
m=\frac{N\pi}{2t\theta},
\end{aligned} 
\end{equation}
where $N=1,2,3,...$. This is a desirable property in QAE. Specifically, it indicates that by adjusting the amplification factor $m$, we can reduce the variance caused by measurements in a quantum computer. In particular, when $m$ satisfies the condition stated in Eq.~\eqref{mg}, the variance can be minimized with a finite number of measurements. For numerical simulation, we utilize the 4-qubit state to demonstrate the efficiency of our algorithm. 

In Fig.~\ref{fig:4q-shot noise}, we show how to minimize the influence of shot noise, i.e., reduce the variance of overlap $\bra{y_-}\mathscr{A}^{mt}\ket{y_-}$ when implementing the measurement, by adjusting the amplification factor $m$. We set the truncation range as $[-1,1]$. The parameter $a$ is fixed as $\frac{1}{10 \sqrt{2}}$. Here, we set the $\theta=3/2$ and the times of measurement $N_{shot}=100$. As shown in the figure, when the magnification factor satisfies Eq.~\eqref{mg}, the variance of the final result becomes minimal.

\subsection{Error from circuit noise}
In addition to cutoff and shot noise, circuit-level noise represents a critical source of error in real quantum computations. This section evaluates the impact of two-qubit gate noise on our method using numerical simulations conducted with the QuESTlink simulator~\cite{jones2020questlink}. To model the circuit noise, we neglect single-qubit gate errors and apply a depolarizing noise channel after each two-qubit gate:
\begin{equation}
\Phi(\rho) = (1 - \epsilon) U \rho U^\dagger + \epsilon \frac{I}{4},
\end{equation}
where $\epsilon$ is the error rate per gate, $U$ is the ideal two-qubit gate, and $I$ is the identity operator on a two-qubit space.

The overall noise depends strongly on the circuit depth, particularly the length of the state preparation operators $U_\psi$ and $U_\phi$ and the number of repeated applications of the amplification operator $\mathscr{A}^m$. Deeper circuits accumulate more noise, leading to diminished signal intensity in $S(x)$.

Our method, like that in *Algorithmic Shadow Spectroscopy*~\cite{chan2025algorithmic}, uses classical Fourier transforms to extract amplitude information from quantum signals. That work has already shown that the total gate noise must remain within a reasonable range—quantified by a parameter $\xi = \lambda \cdot N_{\text{gates}} \approx 1$—to retain usable spectral information. Therefore, our simulation results are consistent with their findings and not unexpected.

In our numerical simulation (see Fig.~\ref{fig:4q-circuit-noise}), we investigate the influence of varying gate error rates on a 4-qubit system. The exact peak location is set to $x = 26.575$, the amplification factor $m=5$, the smoothing parameter $a=\frac{1}{20 \sqrt{2}}$, and the summation range is truncated from $[-\infty,\infty]$ to $[-40,40]$. The results show that while the peak intensity decreases with increasing $\epsilon$, the peak position remains unchanged, highlighting the robustness of our algorithm against moderate circuit noise.

\begin{figure}[t]
\centering
\includegraphics[width=0.95\hsize]{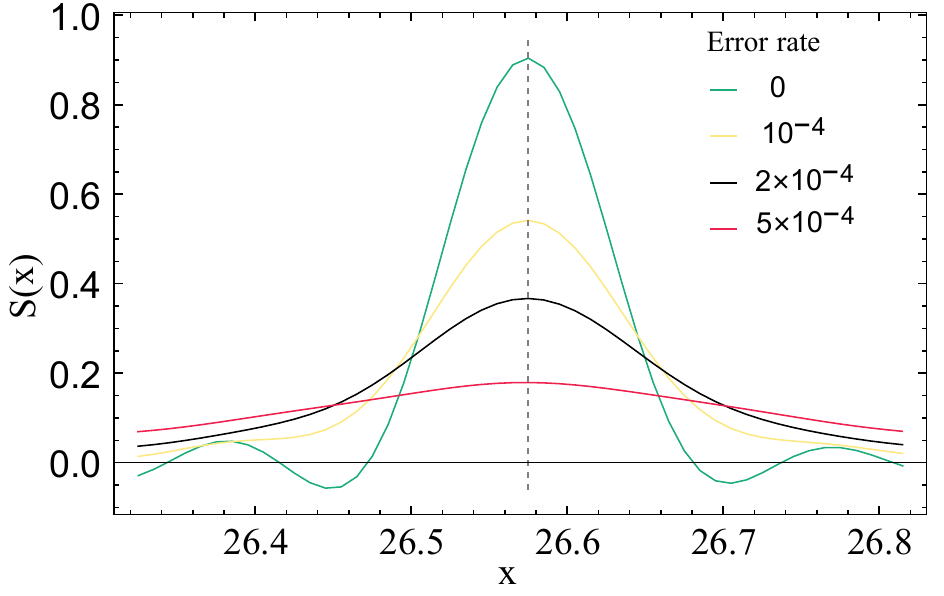}
\caption{The function $S(x)$ evaluated under different two-qubit depolarizing error rates. The gray dashed line marks the exact value $x = 26.575$. While increasing noise reduces the peak height, the position remains stable.}
\label{fig:4q-circuit-noise}
\end{figure}

\section{conclusion}

In this study, we proposed a classical post-processing approach for quantum amplitude estimation that significantly reduces the reliance on extensive quantum resources. Unlike conventional methods, which require numerous ancillary qubits and controlled unitary operations, our method utilizes a sequence of quantum-generated signals that are processed using classical techniques. This approach not only simplifies the overall computational effort but also makes more efficient use of available quantum hardware.

Our numerical simulations demonstrated the effectiveness of this method, showing that it can achieve high accuracy in amplitude estimation with only a fraction of the quantum gates typically required. These results, detailed in Section III, highlight the practical applicability of our method across a range of quantum states, confirming its potential to advance fields such as quantum chemistry and machine learning. Furthermore, our error analysis, outlined in Section IV, suggests that the hybrid approach maintains high precision under various simulation conditions, reinforcing its robustness and scalability. This analysis provides a foundational understanding that will guide future research on optimizing the algorithm for even broader applications in practical quantum computing environments.

Looking ahead, we plan to extend this method to more complex quantum systems and further integrate it with emerging quantum technologies. By continuing to refine the balance between quantum and classical processing, we aim to develop more versatile and powerful tools for the field of quantum computing.

\begin{acknowledgments}
This work is supported by Yunnan Fundamental Research Projects(grant NO.202401BE070001-018).
\end{acknowledgments}

\bibliography{ref.bib}
\onecolumngrid
\appendix

\section{Derivation Process of Eq.\eqref{eq:sx}}
\label{sec A}

This section details the derivation process of Eq.\eqref{eq:sx}.
\begin{equation} \label{a1}
\begin{aligned}
S(x)&=\sum_{t=-\infty}^\infty e^{-a^2t^2}\Tr[\mathscr{A}^{mt}\ket{\psi_0}\bra{\psi_0}]e^{i x t}  \\
&=\Tr[\sum_{t=-\infty}^\infty \mathscr{A}^{mt}\ket{\psi_0}\bra{\psi_0}e^{-a^2t^2} e^{i x t}]. 
\end{aligned}
\end{equation}
Here, we use the linearity of the trace operation, i.e., $\sum_i \Tr[A_i] = \Tr[\sum_i A_i]$ for any set of operators $A_i$. Based on the completeness relation, the operator $\mathscr{A}^{mt}$ can be expanded as
\begin{equation} 
\begin{aligned}
\mathscr{A}^{mt}&=\mathscr{A}^{mt} (\ket{y_+}\bra{y_+}+\ket{y_-}\bra{y_-}) \\
&=e^{2imt\theta}\ket{y_+}\bra{y_+}+e^{-2imt\theta}\ket{y_-}\bra{y_-}.
\end{aligned}
\end{equation}
Substituting this decomposition into Eq.~\eqref{a1}, we obtain
\begin{equation} 
\begin{aligned}
S(x)&=\Tr[\sum_{t=-\infty}^\infty (e^{2imt\theta}\ket{y_+}\bra{y_+}+e^{-2imt\theta}\ket{y_-}\bra{y_-})\ket{\psi_0}\bra{\psi_0}e^{-a^2t^2} e^{i x t}] \\
&=\Tr[\sum_{t=-\infty}^\infty (\ket{y_+} \braket{y_+|\psi_0}\bra{\psi_0}e^{-a^2t^2+i(x+2m\theta)t}+\ket{y_-} \braket{y_-|\psi_0}\bra{\psi_0}e^{-a^2t^2+i(x-2m\theta)t})] \\
&=\sum_{t=-\infty}^\infty ( \Tr[\ket{y_+} \braket{y_+|\psi_0}\bra{\psi_0}] e^{-a^2t^2+i(x+2m\theta)t}+\Tr[\ket{y_-} \braket{y_-|\psi_0}\bra{\psi_0}]e^{-a^2t^2+i(x-2m\theta)t}).
\end{aligned}
\end{equation}
Here, we utilize the equation
\begin{equation} 
\begin{aligned}
\Tr[\ket{y_+} \braket{y_+|\psi_0}\bra{\psi_0}]=\Tr[\ket{y_-} \braket{y_-|\psi_0}\bra{\psi_0}]=\frac{1}{2},
\end{aligned}
\end{equation}
where we have assumed the initial state $\ket{\psi_0} = \frac{1}{\sqrt{2}}(\ket{y_+}+\ket{y_-})$ as defined in Eq.\eqref{ini}. Therefore, Eq.\eqref{a1} simplifies to 
\begin{equation} \label{a5}
\begin{aligned}
S(x)&=\frac{1}{2}\sum_{t=-\infty}^\infty e^{-a^2t^2+i(x+2m\theta)t}+\frac{1}{2}\sum_{t=-\infty}^\infty e^{-a^2t^2+i(x-2m\theta)t} \\
&=\frac{\sqrt{\pi}}{2a}(e^{\frac{-(x+2m\theta)^2}{4a^2}}+e^{\frac{-(x-2m\theta)^2}{4a^2}}).
\end{aligned}
\end{equation}

\section{Cutoff error}
\label{sec: cut} 
To evaluate the impact of truncating the infinite summation in Eq.~\eqref{a5}, we restrict the summation range from $[-\infty, \infty]$ to $[-T, T]$. Specifically, we define
\begin{equation} 
\begin{aligned}
S(x)^\infty=\frac{1}{2}\sum_{t=-\infty}^\infty e^{-a^2t^2+i(x+2m\theta)t}+\frac{1}{2}\sum_{t=-\infty}^\infty e^{-a^2t^2+i(x-2m\theta)t} \\
S(x)^T=\frac{1}{2}\sum_{t=-T}^T e^{-a^2t^2+i(x+2m\theta)t}+\frac{1}{2}\sum_{t=-T}^T e^{-a^2t^2+i(x-2m\theta)t} 
\end{aligned}
\end{equation}
The absolute error introduced by this cutoff is then given by 
\begin{equation} \label{b2}
\begin{aligned}
&|S(x)^\infty - S(x)^T| \\
&=\frac{1}{2}\sum_{t=-\infty}^\infty e^{-a^2t^2+i(x+2m\theta)t}-\frac{1}{2}\sum_{t=-T}^T e^{-a^2t^2+i(x+2m\theta)t}+\frac{1}{2}\sum_{t=-\infty}^\infty e^{-a^2t^2+i(x-2m\theta)t}-\frac{1}{2}\sum_{t=-T}^T e^{-a^2t^2+i(x-2m\theta)t} \\
&=\frac{1}{2}\sum_{t=-\infty}^{-T} e^{-a^2t^2+i(x+2m\theta)t}+\frac{1}{2}\sum_{t=T}^\infty e^{-a^2t^2+i(x+2m\theta)t}+\frac{1}{2}\sum_{t=-\infty}^{-T} e^{-a^2t^2+i(x-2m\theta)t}+\frac{1}{2}\sum_{t=T}^\infty e^{-a^2t^2+i(x-2m\theta)t}
\end{aligned}
\end{equation}
Using the inequality $\sum_{t} e^{-a^2t^2 + i(x \pm 2m\theta) t} \leq \sum_{t} e^{-a^2t^2}$, we can bound the error as

\begin{equation} 
\begin{aligned}
&|S(x)^\infty - S(x)^T| \leq \sum_{t=-\infty}^{-T} e^{-a^2t^2}+\sum_{t=T}^\infty e^{-a^2t^2} \\
&=\frac{2}{a}\rm{erfc}(aT).
\end{aligned}
\end{equation}

Here, $\text{erfc}(x)$ denotes the complementary error function, satisfying the inequality $\rm{erfc}(x)\leqslant e^{-x^2}$. Thus, a further bound can be derived: $|S(x)^\infty - S(x)^T|\leqslant\frac{2}{a}e^{-a^2T^2}$. Therefore, if we want to constrain the error by $\epsilon_c$, the summation range should be no smaller than $\frac{1}{a}\sqrt{\rm{ln\frac{2}{a\epsilon_c}}}$.

\end{document}